\documentclass[journal]{IEEEtran}

\usepackage{amsmath}
\usepackage{amssymb}  
\usepackage{amsthm}	

\usepackage{graphicx} 
\usepackage{enumitem}

\usepackage[hidelinks,bookmarksnumbered=true]{hyperref}
\usepackage{bookmark}
\usepackage{mathtools}
\mathtoolsset{showonlyrefs,showmanualtags}

\usepackage{placeins}
\usepackage{epstopdf}
\usepackage{nameref} 
\usepackage{float}
\usepackage{subfigure}

\usepackage{pgfplots} 

\usepackage{multirow}
\usepackage{array}

\newcolumntype{x}[1]{%
>{\centering\hspace{0pt}}p{#1}}%

\newcommand{\SortAt}[1]{}

\newcommand*\circled[1]{\tikz[baseline=(char.base)]{\node[shape=circle,draw,inner sep=0.5pt] (char) {#1};}}

\theoremstyle{plain}
\newtheorem{theorem}{Theorem}
\newtheorem{lemma}{Lemma}
\newtheorem{corollary}{Corollary}
\newtheorem{remark}{Remark}
\newtheorem{definition}{Definition}

\newcommand{\figref}[1]{Figure~\ref{#1}}
\newcommand{\tabref}[1]{Table~\ref{#1}}

\newcommand{\secref}[1]{Section~\ref{#1}}

\newcommand{\thmref}[1]{Theorem~\ref{#1}}
\newcommand{\lemref}[1]{Lemma~\ref{#1}}
\newcommand{\corref}[1]{Corollary~\ref{#1}}

\newcommand{\procref}[1]{Procedure~\ref{#1}}
\newcommand{\algoref}[1]{Algorithm~\ref{#1}}

\newtheoremstyle{custom}
{\topsep}   
{\topsep}   
{\small}  
{0pt}       
{\normalsize\bfseries} 
{}          
{3pt plus 1pt minus 1pt} 
{\thmname{#1}\thmnumber{ #2}.\thmnote{ (#3)}} 

\usepackage{needspace}

\theoremstyle{custom}
\newtheorem{mytheo}{Procedure}
\newenvironment{procedure}
{%
	\Needspace*{1\baselineskip}%
	\vspace{-1.5mm}\par\noindent\hrulefill \vspace{-2mm}\begin{mytheo}}
	{\vspace{-2mm} \par\noindent\hrulefill\vspace{-1.5mm}
\end{mytheo}} 

\newtheorem{mytheo2}{Algorithm}
\newenvironment{algorithm}
{%
	\Needspace*{1\baselineskip}%
	\vspace{-1.5mm}\par\noindent\hrulefill \vspace{-2mm}\begin{mytheo2}}
	{\vspace{-2mm} \par\noindent\hrulefill\vspace{-1.5mm}
\end{mytheo2}}     

\title{Multivariable Iterative Learning Control Design Procedures: From Decentralized to Centralized, Illustrated on an Industrial Printer}

\author{Lennart Blanken, Tom Oomen%
\thanks{The authors are with the Department of Mechanical Engineering, Eindhoven University of Technology, 5600 MB Eindhoven, The Netherlands (e-mail: l.l.g.blanken@tue.nl; t.a.e.oomen@tue.nl). This work is supported in part by Oc\'e Technologies, and in part by the Netherlands Organisation for Scientific Research (NWO) through research programme VIDI under project 15698.}%
}

\begin{document}
\graphicspath{{figures/}}

\maketitle
\thispagestyle{empty}
\pagestyle{empty}


\begin{abstract}
Iterative Learning Control (ILC) enables high control performance through learning from measured data, using only limited model knowledge in the form of a nominal parametric model to guarantee convergence.
The aim of this paper is to develop a range of approaches for multivariable ILC, where specific attention is given to addressing interaction. The proposed methods either address the interaction in the nominal model, or as uncertainty, i.e., through robust stability. The result is a range of techniques, including the use of the structured singular value (SSV) and Gershgorin bounds, that provide a different trade-off between modeling requirements, i.e., modeling effort and cost, and achievable performance.
This allows an appropriate choice in view of modeling budget and performance requirements.
The trade-off is demonstrated in a case study on an industrial printer.

\end{abstract}


\section{Introduction}\label{MIMO_ILC:sec:introduction}

Iterative Learning Control (ILC) can significantly improve the control performance of systems that perform repeating tasks. After each repetition, or trial, the control action is improved by learning from past trials using an approximate model of the system.
Many successful applications have been reported, including additive manufacturing \cite{Barton2011}, microscopic imaging \cite{Clayton2009}, printing systems \cite{Zundert2016_LQ}, and wafer stages \cite{DeRoover2000}.

The observation that many ILC applications are inherently multivariable has lead to developments of ILC theory for multivariable systems.
Most design algorithms for multivariable ILC have been developed in the so-called lifted or supervector framework \cite{Moore1993}, where the ILC controller follows from a norm-optimization problem over a finite-time horizon, see, e.g., \cite{Owens2013}.

Robust convergence properties of ILC algorithms, i.e., robust stability in trial-domain, are crucial to deal with modeling errors.
Optimization-based algorithms have been further extended to address robust stability in, e.g., \cite{DeRoover2000,Li2016,Son2015,Owens2016,Janssens2013}. These approaches rely on a detailed specification of the nominal model and its uncertainty in a certain prespecified form. Despite being very systematic, this imposes a large burden on the model requirements, since modeling of uncertainty often requires substantial effort of the user \cite{Hjalmarsson2005}. Alternatively, fully data-driven ILC algorithms have been developed in, e.g., \cite{BolderKleOomen2018}, but these require a high experimental cost.

Although robust multivariable ILC has been significantly developed, especially from a theoretical perspective, these approaches are often not employed due to high requirements on uncertainty modeling.
The aim of the present paper is to develop a range of user-friendly multivariable ILC design approaches.
Indeed, in many applications, ILC controllers are designed in the frequency-domain \cite{Moore1993}.
Compared to the norm-optimal framework, this enables a systematic and inexpensive robust design in the sense of modeling requirements, especially regarding the uncertainty. Accurate and inexpensive frequency response function (FRF) measurements \cite{Pintelon2012} can be employed to model the uncertainty, see \cite{DeRoover2000,Strijbosch2018}. In addition, frequency-domain design allows for manual loop-shaping, which is often preferred by control engineers.
However, since such design approaches are mainly single-input single-output (SISO), design for multiple-input multiple-output (MIMO) systems typically involves their application to multiple SISO loops, see, e.g., \cite{Moore1993,Wallen2008}. Interaction is typically ignored, which can lead to stability issues, i.e., non-converging algorithms.
This is especially crucial for ILC, since its control action is effective up to the Nyquist frequency \cite{Paszke2013}.

The seemingly drastical increase in required modeling effort to enforce robust convergence of multivariable ILC algorithms must be justified by the imposed performance requirements.
Interestingly, interaction is typically addressed through full MIMO, or centralized, ILC design. Successful MIMO design approaches include $\mathcal{H}_\infty$ synthesis \cite{DeRoover2000,Zundert2018Mech}, the more restricted class of P-type ILC \cite{Fang1998}, and gradient-based algorithms for point-to-point tracking \cite{Dinh2014}.
Centralized techniques enable robust convergence and superior performance, yet require a MIMO parametric model of the system, including interaction. These models can be difficult and expensive to obtain, especially for lightly damped mechatronic systems due to complex dynamics and numerical issues \cite{Oomen2018}.

The main contribution of this paper is a systematic design framework for analysis and synthesis of multivariable ILC, that explicitly addresses the design trade-offs between modeling and performance requirements. The proposed solutions, which form subcontributions, range from decentralized to centralized designs, with various levels of modeling requirements.
The decentralized designs build on results in, e.g., \cite{Grosdidier1986} to guarantee robust convergence, including the use of the structured singular value \cite[Chapter 11]{Zhou1996}, and require limited user effort using only SISO parametric models.
The effectiveness of the framework is demonstrated in a case study on an industrial flatbed printer.
The paper extends preliminary results in \cite{Blanken2016CDC,Blanken2016Mech} through the design framework, new technical results, detailed proofs, and application results.

\textit{Notation.} The imaginary unit is denoted $\iota$, i.e., $\iota^2=-1$.

\section{Problem Formulation}\label{MIMO_ILC:sec:problem_formulation}

\subsection{ILC Setup}\label{MIMO_ILC:subsec:ILC_setup}
\begin{figure}[tpb]
	\centering
	\mbox{\includegraphics[scale=0.85,page=2]{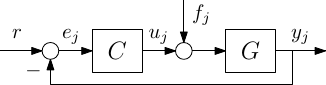}}
	\caption{ILC control configuration.}
	\label{MIMO_ILC:fig:block_scheme_cloop}
\end{figure}
Consider the control configuration in \figref{MIMO_ILC:fig:block_scheme_cloop}, consisting of possibly non-square plant $G\in\mathcal{R}^{n_y\times n_u}(z)$ and internally stabilizing feedback controller $C\in\mathcal{R}^{n_u\times n_y}(z)$.
The disturbance $r\in\ell_2$ is trial-invariant, where each repetition, or trial, is denoted by index $j\in\mathbb{N}_{\geq0}$.
The aim is to minimize the tracking error $e$ in the presence of $r$. Note that trial-varying disturbances are tacitly omitted, see, e.g., \cite{Gunnarsson2006,Oomen2017Mech} for details. The output in trial $j$ is denoted $y_j$, the feedforward by $f_j$, and
\begin{equation}\label{MIMO_ILC:eq:ILC_error}
e_j = Sr-J f_j,
\end{equation}
with sensitivity function $S=(I+GC)^{-1}\in\mathcal{RH}_\infty^{n_y\times n_y}$, and process sensitivity function $J=SG\in\mathcal{RH}_\infty^{n_y\times n_u}$.
Zero initial conditions are assumed without loss of generality \cite{Moore1993}. If $G$ is stable, then $C=0$ is admissible such that $S=I$ and $J=G$.

The objective of ILC is to improve control performance in the next trial $j+1$ by selecting the command input $f_{j+1}$. Typically, an algorithm of the following form is invoked:
\begin{equation}\label{MIMO_ILC:eq:ILC_update}
f_{j+1}= Q ( f_j + L e_j ),
\end{equation}
where $L\in\mathcal{RL}_\infty^{n_u\times n_y}$ and $Q\in\mathcal{RL}_\infty^{n_u\times n_u}$ are to be designed. Notice $L$, $Q$ can be non-causal, since \eqref{MIMO_ILC:eq:ILC_update} is computed off-line.

\subsection{ILC Design for SISO Systems}\label{MIMO_ILC:subsec:SISO_design_procedure}
For the case $n_u=n_y=1$, design procedures are well developed. Often a two-step approach is used, see, e.g., \cite{Strijbosch2018}.
\begin{procedure}\label{MIMO_ILC:proc_SISO_ILC}
	{\normalsize Frequency-domain SISO ILC design}\vspace{-2mm}
	\par\noindent\hrulefill
	\begin{enumerate}
		\item Choose $L(z)$ such that $L(e^{\iota\omega})J(e^{\iota\omega})\approx1$, $\omega\in[0,\omega_c]$.
		This step requires a parametric model of $J(z)$.
		\item For robust stability, $Q(z)$ is selected as a low-pass filter with cut-off frequency near $\omega_c$, such that $Q(e^{\iota\omega})\approx0$, $\forall \omega>\omega_c$.
		This can be performed using nonparametric models of $J(e^{\iota\omega})$.
	\end{enumerate}
\end{procedure}
\procref{MIMO_ILC:proc_SISO_ILC} requires limited model knowledge, since robust stability can be guaranteed through FRFs, see, e.g., \cite{Pintelon2012}, which are for mechatronic systems often accurate and fast to obtain.

A naive extension of \procref{MIMO_ILC:proc_SISO_ILC} to the multivariable case could be to implement multiple SISO ILC.
In this paper, it is demonstrated that this can lead to non-convergent algorithms.

\subsection{Problem Formulation and Contributions}
The problem considered in this paper is the design of multivariable filters $L(z)$ and $Q(z)$ in \eqref{MIMO_ILC:eq:ILC_update} in the frequency domain with respect to the following requirements:
\begin{enumerate}[leftmargin=27pt]
	\renewcommand{\labelenumi}{\theenumi)}
	\renewcommand{\theenumi}{R\arabic{enumi}}
	\item \label{MIMO_ILC:R1} Robust convergence of \eqref{MIMO_ILC:eq:ILC_update}, i.e., stability in trial domain;
	\item \label{MIMO_ILC:R2} High control performance, i.e., a small error $e_j$;
	\item \label{MIMO_ILC:R3} Limited required user effort.
\end{enumerate}
The term \textit{user effort} relates to design tools and required models, i.e., parametric vs. nonparametric, and SISO vs. MIMO.

The main contribution is the development of a step-by-step design procedure for multivariable iterative learning control that addresses modeling and robustness aspects. The proposed design techniques vary in sophistication, and range from
\begin{itemize}
	\item decentralized designs, using SISO parametric models, to
	\item centralized designs, requiring MIMO parametric models,
\end{itemize}
where in all cases, robustness to modeling errors is addressed using nonparametric FRF measurements.
The procedure generalizes \procref{MIMO_ILC:proc_SISO_ILC} to the MIMO case, and provides a coherent overview of available approaches, such that a well-motivated choice can be made for the problem at hand in view of \ref{MIMO_ILC:R1}-\ref{MIMO_ILC:R3}.

\subsection{Overview of Design Framework and Outline of Paper}
The present paper addresses theoretical, design, and algorithmic aspects to obtain a practically implementable design framework for MIMO ILC.
The framework connects all design approaches, see \figref{MIMO_ILC:fig:overview_approaches}, and is summarized next.
\begin{figure}[tpb]
		\centering
		\mbox{\includegraphics[scale=0.85,page=2]{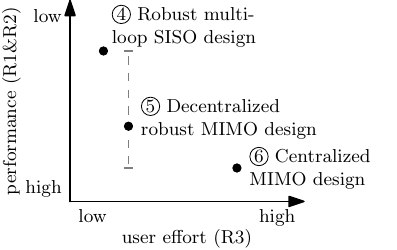}}
		\caption{Schematic overview of approaches in the design framework, illustrating the trade-offs between performance (\ref{MIMO_ILC:R1} and \ref{MIMO_ILC:R2}), and user effort (\ref{MIMO_ILC:R3}). Depending on the level of interaction, the vertical position of \protect\circled{5} may vary.}
		\label{MIMO_ILC:fig:overview_approaches}
\end{figure}

\begin{procedure}\label{MIMO_ILC:proc_MIMO_ILC}
	{\normalsize Frequency-domain MIMO ILC design}\vspace{-2mm}
	\par\noindent\hrulefill
	\begin{enumerate}
		[label=\protect\circled{\arabic*}]
		\item \underline{Non-parametric modeling}.
		\begin{itemize}
			\item Identify MIMO FRF model $\hat{J}_\textrm{FRF}(e^{\iota\omega})$ of $J$, see, e.g., \cite{Pintelon2012}.
		\end{itemize} 
		\item \underline{Interaction analysis}. Decoupled?
		\begin{itemize}
			\item If yes: independent SISO design (\procref{MIMO_ILC:proc_SISO_ILC}).
		\end{itemize}
		\item \underline{Decoupling transformations}. Decoupled?
		\begin{itemize}
			\item If yes: independent SISO design (\procref{MIMO_ILC:proc_SISO_ILC}).
		\end{itemize}	
		\item \underline{Robust multi-loop SISO design} (\secref{MIMO_ILC:sec:multi-loop_SISO}).
		\begin{enumerate}[label=\roman*)]
			\item Obtain SISO parametric models $\widehat{J_{ii}}(z)$;
			\item Robust multi-loop SISO design of $L(z)=\mathrm{diag}\{L_{ii}(z)\}$ and $Q(z)$ using $\widehat{J_{ii}}(z)$ and FRF model $\hat{J}_\textrm{FRF}(e^{\iota\omega})$ (\algoref{MIMO_ILC:algo_robust_multi-loop_ILC}).
		\end{enumerate}
		Performance not satisfactory? Proceed to next step.
		\item \underline{Decentralized robust MIMO design} (\secref{MIMO_ILC:sec:decentralized}).
		\begin{enumerate}[label=\roman*)]
			\item Decentralized design of $L(z)$ and $Q(z)$ for robustness to deliberately ignored interaction and modeling errors, using $\widehat{J_{ii}}(z)$ and FRF model $\hat{J}_\textrm{FRF}(e^{\iota\omega})$ (\algoref{MIMO_ILC:algo_decentralized_ILC}).
		\end{enumerate}
		Performance not satisfactory? Proceed to next step.
		\item \underline{Centralized MIMO design} (\secref{MIMO_ILC:sec:centralized}).
		\begin{enumerate}[label=\roman*)]
			\item Obtain MIMO parametric model $\hat{J}(z)$, including interaction;
			\item MIMO design of $L(z)$ and $Q(z)$ for robustness to modeling errors, using $\hat{J}(z)$ and FRF model $\hat{J}_\textrm{FRF}(e^{\iota\omega})$ (\algoref{MIMO_ILC:algo_centralized_ILC}).
		\end{enumerate}
	\end{enumerate}
\end{procedure}
The key point is that modeling requirements should only be increased if justified by performance requirements. Indeed, $\circled{2}$ to $\circled{5}$ require only SISO parametric models and an FRF measurement, and may yield satisfactory performance, while \circled{6} requires a costly MIMO parametric model.

\begin{remark}
	Plant uncertainty can directly be addressed in \procref{MIMO_ILC:proc_MIMO_ILC} through confidence intervals of FRF estimates.
\end{remark}

The outline of the paper is as follows.
First, the design problem is analyzed.
Then, in Sections \ref{MIMO_ILC:sec:multi-loop_SISO} to \ref{MIMO_ILC:sec:centralized}, the design techniques are developed that constitute steps \circled{4} to \circled{6}.
In \secref{MIMO_ILC:sec:case_study}, \procref{MIMO_ILC:proc_MIMO_ILC} is applied to a multivariable case study. A preview on the results is presented in \figref{MIMO_ILC:fig:norm_error}, illustrating the trade-offs between approaches in \procref{MIMO_ILC:proc_MIMO_ILC}.

\begin{figure}[tpb]
	\centering
	\mbox{\includegraphics{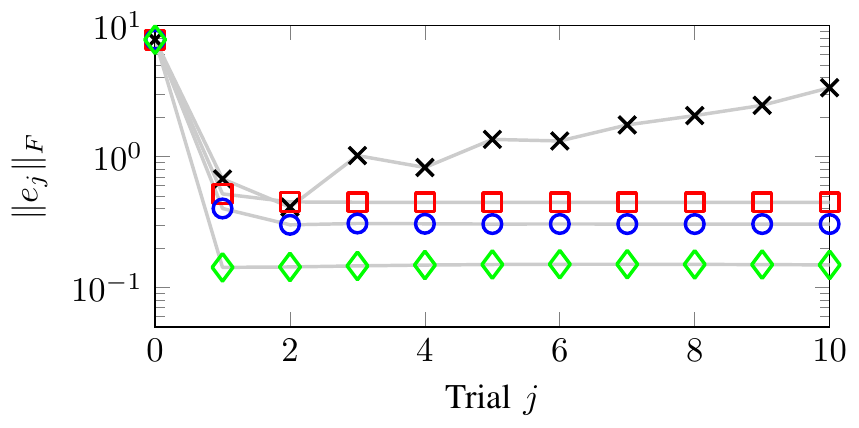}}
	\caption{Simulation results: often-used multi-loop SISO ILC design (\protect\tikz[baseline=-0.6ex,x=1pt,y=1pt]{\protect\draw[black,thick] [-] (0,-3) -- (6,3);\protect\draw[black,thick] [-] (0,3) -- (6,-3);}) can lead to non-convergent algorithms. Through the developed design approaches in steps {\protect\circled{4}} (\protect\tikz[baseline=-0.6ex,x=1pt,y=1pt]{\protect\draw[red,thick] (0,-3) rectangle ++(6,6);}), {\protect\circled{5}} (\protect\tikz[baseline=-0.6ex,x=1pt,y=1pt]{\protect\draw[blue,thick] (7,0) circle (3);}), and {\protect\circled{6}} (\protect\tikz[baseline=-0.6ex,x=1pt,y=1pt]{\protect\draw[green,thick] (0,0) +(3,0) -- +(0,3) -- +(-3,0) -- +(0,-3) -- cycle;}), a well-motivated balance can be made between achievable performance, i.e., norm of the asymptotic error $\|e_\infty\|_F$, and the associated user effort in terms of modeling cost and design complexity.}
	\label{MIMO_ILC:fig:norm_error}
\end{figure}

\section{Analysis of ILC Design Problem}\label{MIMO_ILC:sec:analysis_ILC}
In this section, the general ILC algorithm \eqref{MIMO_ILC:eq:ILC_update} is analyzed, and robust convergence and control performance are defined.

\subsection{Convergence and Performance}\label{MIMO_ILC:subsec:convergence_analysis}
Combining \eqref{MIMO_ILC:eq:ILC_error} and \eqref{MIMO_ILC:eq:ILC_update} yields the linear iterative systems that describe the propagation of $e_j$ and $f_j$ in the trial domain:
\begin{align}
f_{j+1} &= Q(I-LJ)f_j + QLSr,\label{MIMO_ILC:eq:ILC_feedforward_propagation}\\
e_{j+1} &= JQ(I-LJ)J_l^{-1}e_j + (I-JQJ_l^{-1})Sr.\label{MIMO_ILC:eq:ILC_error_propagation}
\end{align}
where \eqref{MIMO_ILC:eq:ILC_error_propagation} holds if a left inverse $J_l^{-1}$ exists such that $J_l^{-1}J=I$, i.e., at least $n_y\geq n_u$.
Convergence is formalized next.
\begin{definition}\label{MIMO_ILC:def:convergence}
	System \eqref{MIMO_ILC:eq:ILC_feedforward_propagation} is convergent iff for all $r,f_0\in\ell_2$, there exists an asymptotic signal $f_\infty\in\ell_2$ such that 
	\begin{equation}
	\limsup_{j\rightarrow\infty}\|f_\infty-f_j\|=0.
	\end{equation}
\end{definition}
Then, the asymptotic signals $f_\infty$ and $e_\infty$ are obtained as
\begin{align}
f_\infty &= \lim_{j\rightarrow\infty}f_j = \left(I-Q(I-LJ)\right)^{-1}QLSr ,\label{MIMO_ILC:eq:fixed_point_ff}\\
e_\infty &= \lim_{j\rightarrow\infty}e_j = \left(I-J\left(I-Q(I-LJ)\right)^{-1}QL\right)Sr.\label{MIMO_ILC:eq:fixed_point_error}
\end{align}
A condition for convergence of \eqref{MIMO_ILC:eq:ILC_feedforward_propagation} and \eqref{MIMO_ILC:eq:ILC_error_propagation} is given next.
\begin{theorem}\label{MIMO_ILC:thm:conv_ILC}
	Iterations \eqref{MIMO_ILC:eq:ILC_feedforward_propagation}, \eqref{MIMO_ILC:eq:ILC_error_propagation} converge to $f_\infty$ and $e_\infty$ iff
	\begin{equation}\label{MIMO_ILC:eq:conv_cond}
	\rho\left(Q(e^{\iota\omega})(I-L(e^{\iota\omega})J(e^{\iota\omega})\right) < 1 \quad \forall\omega\in[0,\pi],
	\end{equation}
	where $\rho(\cdot)$ denotes spectral radius, i.e., $\rho(\cdot)=\max_i|\lambda_i(\cdot)|$.
\end{theorem}
See \cite[Theorem 6]{Norrlof2002} for a proof, which can be appropriately extended for non-causal $L,Q\in\mathcal{RL}_\infty$.
Although \eqref{MIMO_ILC:eq:conv_cond} guarantees convergence, it does not guarantee good learning transients.
Monotonic convergence is considered next.
\begin{definition}\label{MIMO_ILC:def:mon_convergence}
	Iteration \eqref{MIMO_ILC:eq:ILC_feedforward_propagation} converges monotonically w.r.t. the $\ell_2$ norm of $f_j$ to $f_\infty$ with convergence rate $\gamma$, $0\leq\gamma<1$, iff
	\begin{align}
	\|f_\infty-f_{j+1}\|_{\ell_2} \leq \gamma\|f_\infty-f_j\|_{\ell_2}\quad\forall j.
	\end{align}
\end{definition}
\begin{theorem}\label{MIMO_ILC:thm:mon_conv_ILC}
	Iteration \eqref{MIMO_ILC:eq:ILC_feedforward_propagation} converges monotonically w.r.t. the $\ell_2$ norm of $f_j$ to fixed point $f_\infty$, with convergence rate $\gamma$, iff
	\begin{equation}\label{MIMO_ILC:eq:mon_conv_cond}
	\gamma := \left\|Q(I-LJ)\right\|_\infty < 1,
	\end{equation}
	where
	$\|H(z)\|_\infty=\sup_{\omega\in[0,\pi]}\bar{\sigma}(H(e^{\iota\omega}))$ is the $\mathcal{L}_\infty$-norm, and $\bar{\sigma}$ denotes the maximum singular value. 
\end{theorem}
See \cite[Theorem 2]{Oomen2017Mech} for a proof.
Note \eqref{MIMO_ILC:eq:mon_conv_cond} is equivalent to
\begin{equation}\label{MIMO_ILC:eq:mon_conv_cond_sigma}
\bar{\sigma}\left(Q(e^{\iota\omega})(I-L(e^{\iota\omega})J(e^{\iota\omega}))\right) < 1 \quad \forall\omega\in[0,\pi].
\end{equation}
In view of \eqref{MIMO_ILC:eq:fixed_point_error}, the following result is crucial for performance.
\begin{theorem}\label{MIMO_ILC:thm:mon_conv_ILC_zero_error}
	Assume $L(e^{\iota\omega}),J(e^{\iota\omega})\neq0$, $\forall \omega$, in iteration \eqref{MIMO_ILC:eq:ILC_feedforward_propagation}.
	Given that \eqref{MIMO_ILC:eq:conv_cond} holds, then for all $r\in\ell_2$, $e_\infty=0$ iff $Q=I$.
\end{theorem}
\begin{proof}
	Since \eqref{MIMO_ILC:eq:conv_cond} holds, the fixed points $e_\infty$ and $f_\infty$ exist.\\
	$(\Leftarrow)$ If $Q=I$, then $f_\infty=f_\infty+Le_\infty$, which implies $e_\infty=0$.\\
	$(\Rightarrow)$ If $e_\infty=0$, then $f_\infty=Qf_\infty$, which implies $Q=I$.
\end{proof}

\subsection{Design and Modeling Considerations}\label{MIMO_ILC:subsec:design_modeling_considerations}
Theorems \ref{MIMO_ILC:thm:conv_ILC}, \ref{MIMO_ILC:thm:mon_conv_ILC} and \ref{MIMO_ILC:thm:mon_conv_ILC_zero_error} have direct consequences for the design of $L$ and $Q$. To achieve $e_\infty=0$, $L$ should be designed such that $\rho(I-LJ) < 1$, and $Q$ must equal $I$. For fast convergence, $\left\|Q(I-LJ)\right\|_\infty$ should be small.

\section{Multi-Loop SISO ILC Design}\label{MIMO_ILC:sec:multi-loop_SISO}
For systems where interaction is absent or sufficiently small, possibly after a decoupling process, multiple SISO ILC can be designed, see \circled{2} and \circled{3}.
In this section, it is shown that in the presence of interaction, multi-loop SISO designs may lead to non-convergent schemes, i.e., \ref{MIMO_ILC:R1} is not guaranteed. To account for ignored interaction, the ILC can be robustified \textit{a posteriori}, which is shown to lead to conservatism, hence compromising \ref{MIMO_ILC:R2}.
It is assumed that $J$ is square, i.e., $n\equiv n_u=n_y$, possibly after a squaring-down process, see, e.g., \cite{Zundert2018ACC}.

\subsection{Independent SISO ILC Design for MIMO Systems}\label{MIMO_ILC:subsec:SISO_design}
If no coupling is present, i.e., $J(z)=\mathrm{diag}\{J_{ii}(z)\}$, then multi-loop SISO filters $L$ and $Q$ can be designed by application of \procref{MIMO_ILC:proc_SISO_ILC} to each loop $i=1,\ldots,n$.
That is, design
\begin{align}
L = \mathrm{diag}\{l_1,l_2,\ldots,l_n \},\quad
Q = \mathrm{diag}\{q_1,q_2,\ldots,q_n \},\label{MIMO_ILC:eq:LQ_decentralized}
\end{align}
according to the set of SISO criteria
\begin{align}
\left| q_{i}(e^{\iota\omega})(1-l_{i}(e^{\iota\omega})J_{ii}(e^{\iota\omega})) \right| < 1 \quad\forall i, \omega\in[0,\pi].\label{MIMO_ILC:eq:convergence_decentralized}
\end{align}
Typically, each $l_i(z)$ is based on inversion of $\widehat{J_{ii}}(z)$, see, e.g., \cite{Moore1993,Wallen2008}, and \cite{Zundert2018Mech,Butterworth2012} for algorithms.
Often, $q_i(z)$ are zero-phase filters, and are implemented non-causally, i.e., an operation with $q_i$ and its adjoint $q_i^*$, see, e.g., \cite{BolderKleOomen2018,Strijbosch2018}.

However, when considerable interaction is present, independent SISO designs may lead to non-convergent systems, see also \figref{MIMO_ILC:fig:norm_error}.
Indeed, \eqref{MIMO_ILC:eq:convergence_decentralized} does not guarantee \thmref{MIMO_ILC:thm:conv_ILC}.

\subsection{Accounting for Ignored Interaction Through Robustness}\label{MIMO_ILC:subsec:Q_robustness_interaction}
Several approaches can be taken based on Theorems \ref{MIMO_ILC:thm:conv_ILC} and \ref{MIMO_ILC:thm:mon_conv_ILC} to enable SISO design of $Q$ for robust MIMO convergence.
Their restrictiveness is subject to a trade-off with the assumptions made on the structure of $Q$. Selecting $Q(z)=q_d(z)I$ with SISO filter $q_d(z)\in\mathcal{RL}_\infty$ leads to the next result.
\begin{corollary}\label{MIMO_ILC:cor:ILC_conv_qd}
	Assume $Q(z)=q_d(z)I\in\mathcal{RL}_\infty^{n\times n}$ with SISO filter $q_d(z)\in\mathcal{RL}_\infty$. The iteration \eqref{MIMO_ILC:eq:ILC_feedforward_propagation} converges iff
	\begin{equation}\label{MIMO_ILC:eq:conv_cond_qd}
	|q_d(e^{\iota\omega})|\rho\left(I-L(e^{\iota\omega})J(e^{\iota\omega})\right) < 1 \quad \forall\omega\in[0,\pi],
	\end{equation}
	and converges monotonically w.r.t. the $\ell_2$ norm of $f_j$ iff
	\begin{equation}\label{MIMO_ILC:eq:mon_conv_cond_qd}
	|q_d(e^{\iota\omega})|\bar{\sigma}\left(I-L(e^{\iota\omega})J(e^{\iota\omega})\right) < 1 \quad \forall\omega\in[0,\pi].
	\end{equation}
\end{corollary}
\corref{MIMO_ILC:cor:ILC_conv_qd} enables SISO design of $q_d(z)$ that guarantees robust convergence of the MIMO system using $\hat{J}_\textrm{FRF}(e^{\iota\omega})$.
This leads to the following design algorithm, constituting step \circled{4}.
\begin{algorithm}\label{MIMO_ILC:algo_robust_multi-loop_ILC}
	{\normalsize Step \circled{4}: robust multi-loop SISO design}\vspace{-2mm}
	\par\noindent\hrulefill
	\begin{enumerate}
		[label=\alph*)]
		\item Obtain SISO parametric models $\widehat{J_{ii}}(z)$ of $J_{ii}(z)$, $i=1,\ldots,n$.
		\item Design multi-loop SISO learning filter $L(z)=\mathrm{diag}\{l_{i}(z)\}$ such that $1-l_{i}(e^{\iota\omega})\widehat{J_{ii}}(e^{\iota\omega}) \approx 0$,  $\forall i,\omega\in[0,\pi]$.
		\item Based on MIMO non-parametric model $\hat{J}_\textrm{FRF}(e^{\iota\omega})$ from step \circled{1} of \procref{MIMO_ILC:proc_MIMO_ILC}, design $Q(z)=q_d(z)I$ according to \corref{MIMO_ILC:cor:ILC_conv_qd} to guarantee robust stability of the MIMO algorithm.
	\end{enumerate}
\end{algorithm}

To conclude, convergence can be guaranteed. However, performance may be limited: \corref{MIMO_ILC:cor:ILC_conv_qd} is very restrictive on the structure of $Q$. This motivates the development of decentralized designs, where each loop is robustified individually.

\section{Decentralized ILC: Robustness to Interaction Through Independent \texorpdfstring{$Q$}{Q}-filter Designs}\label{MIMO_ILC:sec:decentralized}
For systems where interaction cannot be ignored in view of convergence, yet high performance is desired using only SISO parametric models, decentralized ILCs can be designed, see \circled{5}.
In this section, a decentralized design approach is developed that guarantees robust convergence of the MIMO system. The approach requires the same models as \algoref{MIMO_ILC:algo_robust_multi-loop_ILC}: user effort (\ref{MIMO_ILC:R3}) is only increased by more involved computations.

In this section, the focus is on decentralized design of $Q(z)\in\mathcal{RL}_\infty^{n_u\times n_u}$ for given $L(z)\in\mathcal{RL}_\infty^{n_u\times n_y}$, which can be diagonal or full MIMO, see steps \circled{5} and \circled{6}. 
Yet, the results are foreseen to be most often applied to square systems, since this also enables decentralized design of $L=\mathrm{diag}\{l_1,l_2,\ldots,l_{n}\}$.

\subsection{Factorization of Iteration Dynamics}\label{MIMO_ILC:subsec:dec_ILC_interaction}
To analyze the role of interaction, \eqref{MIMO_ILC:eq:ILC_feedforward_propagation} is factored as
\begin{align}
f_{j+1}=QM f_j + \tilde{w} &= QM_\mathrm{d}(I+E) f_j + \tilde{w} \label{MIMO_ILC:eq:ILC_f_propagation}, 
\end{align}
where $M=I-LJ$, $M_\mathrm{d}=\mathrm{diag}\{M_{ii}\}$ consists of the diagonal elements of $M$, $E=M_\mathrm{d}^{-1}(M-M_\mathrm{d})$ contains the normalized interaction in $M$, and $\tilde{w}=QLSr$, see \figref{MIMO_ILC:fig:block_scheme_ILC_decentralized}.
Note that $M_\mathrm{d}$ and $E$ are functions of $J$ and $L$, and $E=0$ if $J$ is diagonal, i.e., there is no interaction.
The interaction term $I+E$ can be used to analyze robust stability. The following result is the basis for forthcoming decentralized designs.
\begin{lemma}\label{MIMO_ILC:lem:factorization_rho_conditions}
	Iteration \eqref{MIMO_ILC:eq:ILC_feedforward_propagation} converges iff
	\begin{align}\label{MIMO_ILC:eq:ILC_conv_cond_rho}
	\rho\left(QM_\mathrm{d}(I+E)\right) < 1 \quad\forall \omega\in[0,\pi].
	\end{align}
	and converges monotonically w.r.t. the $\ell_2$ norm of $f_j$ iff
	\begin{align}\label{MIMO_ILC:eq:ILC_mon_conv_cond_rho}
	\rho\left(M_\mathrm{d}^HQ^HQM_\mathrm{d}(I+E)(I+E)^H\right) < 1 \quad\forall \omega\in[0,\pi],\hspace{5pt}
	\end{align}
	where the superscript $H$ denotes conjugate transpose.
\end{lemma}
\begin{proof}
	\eqref{MIMO_ILC:eq:ILC_conv_cond_rho} follows directly from substituting \eqref{MIMO_ILC:eq:ILC_f_propagation} into \eqref{MIMO_ILC:eq:conv_cond}, and
	\eqref{MIMO_ILC:eq:ILC_mon_conv_cond_rho} follows from substituting \eqref{MIMO_ILC:eq:ILC_f_propagation} into \eqref{MIMO_ILC:eq:mon_conv_cond_sigma}, and rewriting:
	\begin{subequations}
		\begin{align}
		\bar{\sigma}(Q(I-LJ&)) = \sqrt{\rho(QM(QM)^H)} = \sqrt{\rho(QMM^HQ^H)} \\
		&= \sqrt{\rho(QM_\mathrm{d}(I+E)(I+E)^HM_\mathrm{d}^HQ^H)} \\
		&= \sqrt{\rho(M_\mathrm{d}^HQ^HQM_\mathrm{d}(I+E)(I+E)^H)},\label{MIMO_ILC:eq:sigma_basis}
		\end{align}
	\end{subequations}
	where it is used that $\{\lambda_i(AB)\}=\{\lambda_i(BA)\}$ for square $A,B$.
	Substituting \eqref{MIMO_ILC:eq:sigma_basis} in \eqref{MIMO_ILC:eq:mon_conv_cond_sigma} and squaring completes the proof.
\end{proof}
Two observations are made for forthcoming developments:
\begin{itemize}
	\item The matrix $QM_\mathrm{d}$ in \eqref{MIMO_ILC:eq:ILC_conv_cond_rho}, respectively $M_\mathrm{d}^HQ^HQM_\mathrm{d}$ in \eqref{MIMO_ILC:eq:ILC_mon_conv_cond_rho}, has diagonal structure and is right multiplied with interaction term $(I+E)$, respectively $(I+E)(I+E)^H$.
	\item Comparing with \eqref{MIMO_ILC:eq:mon_conv_cond_sigma}, condition \eqref{MIMO_ILC:eq:ILC_mon_conv_cond_rho} is based on a spectral radius $\rho(\cdot)$ instead of a maximum singular value $\bar{\sigma}(\cdot)$.
\end{itemize}
Together, the structured form and the use of $\rho(\cdot)$ allow for the development of robust decentralized design techniques.
\begin{figure}[tpb]
		\centering
		{\includegraphics[page=1,scale=0.85]{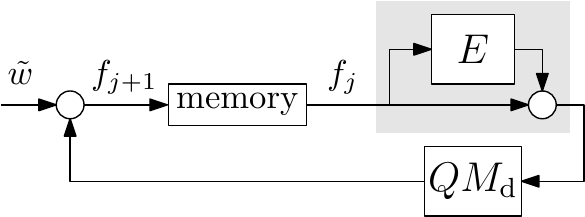}}
		\caption{Schematical representation of the factored iteration \eqref{MIMO_ILC:eq:ILC_f_propagation}. Since the interaction term $I+E$ (grey area) is invariant to $Q$, measures on $I+E$ can be developed to design decentralized filters $Q$ for robust convergence.}
		\label{MIMO_ILC:fig:block_scheme_ILC_decentralized}
\end{figure}

\begin{remark}\label{MIMO_ILC:rem:E_vs_IplusE}
	The factorization \eqref{MIMO_ILC:eq:ILC_f_propagation} resembles decentralized feedback control, see, e.g., \cite{Grosdidier1986} and \cite[Section 10.6]{Skogestad2007}, yet fundamentally differs regarding the use of $E$.
	
	In decentralized feedback design, i.e., $K=\mathrm{diag}\{k_i\}$ with open-loop transfer $GK$, the return difference is factored as
	\begin{equation}\label{MIMO_ILC:eq:FB_return_difference}
	I+GK=(I+E\tilde{T})(I+G_\mathrm{d}K),
	\end{equation}
	where $E=(G-G_\mathrm{d})G_\mathrm{d}^{-1}$, $G_\mathrm{d}=\mathrm{diag}\{G_{ii}\}$, and $\tilde{T}=(I+G_\mathrm{d}K)^{-1}G_\mathrm{d}K$. Assuming that $\tilde{T}$ is stable, the closed-loop $T=I-S$ is stable if $\rho(E\tilde{T})<1$, $\forall\omega\in[0,\pi]$, see \cite[Theorem 2]{Grosdidier1986}.
	Since $E$ appears linearly in $E\tilde{T}$, the magnitude of $E$ w.r.t. $\tilde{T}$ is typically used to analyze stability, see, e.g., \cite{Grosdidier1986}. 
	
	In decentralized ILC, \eqref{MIMO_ILC:eq:ILC_f_propagation} is affine in $E$, and bounds on $I+E$ w.r.t. $QM_\mathrm{d}$ are developed.
	The key difference is that \eqref{MIMO_ILC:eq:ILC_f_propagation} is factored, which is the open-loop in iteration domain, while in feedback the closed-loop return difference \eqref{MIMO_ILC:eq:FB_return_difference} is factored.
\end{remark}

\subsection{Decentralized Conditions for Robust Convergence} \label{MIMO_ILC:subsec:decentralized_robust_ILC_bounds}
Next, several decentralized design conditions are developed.
The conditions are less conservative than \corref{MIMO_ILC:cor:ILC_conv_qd} since the decentralized structure of $Q$ is explicitly taken into account. 

\subsubsection{Independent \texorpdfstring{$Q$}{Q}-filter Design Based on Induced Norms} \label{MIMO_ILC:subsec:dec_ILC_indep_Q_Gershgorin}
In this subsection, upper bounds on the spectral radii in \eqref{MIMO_ILC:eq:ILC_conv_cond_rho} and \eqref{MIMO_ILC:eq:ILC_mon_conv_cond_rho} based on induced norms are used for decentralized design of $Q$. In particular, for any matrix $A$, it holds $\rho(A)\leq \|A\|_{ip}$.
This relation is crucial for the presented designs.
\begin{theorem}\label{MIMO_ILC:thm:ILC_decentralized_Gershgorin}
	Iteration \eqref{MIMO_ILC:eq:ILC_feedforward_propagation} with decentralized filter $Q(z)=\mathrm{diag}\{q_{i}(z)\}$, as in \eqref{MIMO_ILC:eq:LQ_decentralized}, converges if either:
	\begin{align}
		\hspace{-6mm}|q_i(e^{\iota\omega})M_{ii}(e^{\iota\omega})|  &< \tfrac{1}{\vphantom{\tilde{A}}\sum_{j}|I + E(e^{\iota\omega})|_{ij}}\quad \forall i,\omega\in[0,\pi],\label{MIMO_ILC:eq:ILC_conv_iinfty_SISO} \\
		\hspace{-6mm}|q_i(e^{\iota\omega})M_{ii}(e^{\iota\omega})|  &< \tfrac{1}{\vphantom{\tilde{A}}\sum_{j}|I + E(e^{\iota\omega})|_{ji}}\quad \forall i,\omega\in[0,\pi],\hspace{0mm}\label{MIMO_ILC:eq:ILC_conv_i1_SISO}
	\end{align}
	and converges monotonically w.r.t. the $\ell_2$ norm of $f_j$ if
	\begin{align}
		|q_i(e^{\iota\omega})M_{ii}(e^{\iota\omega})|  < \tfrac{1}{\vphantom{\tilde{A}}\sqrt{\sum_{j}|(I + E(e^{\iota\omega}))(I + E(e^{\iota\omega}))^H|_{ij}}} \\ 
		\forall i,\omega\in[0,\pi],\hspace{10pt}\label{MIMO_ILC:eq:ILC_mon_conv_iinfty_SISO}
	\end{align}
	where $|\cdot|$ denotes element-wise absolute value.
\end{theorem}
\begin{proof}
		Conditions \eqref{MIMO_ILC:eq:ILC_conv_iinfty_SISO}, \eqref{MIMO_ILC:eq:ILC_conv_i1_SISO} follow by application of $\rho(A)\leq \|A\|_{ip}$ to \eqref{MIMO_ILC:eq:ILC_conv_cond_rho} with $p=\infty$, $p=1$, respectively. Similarly, \eqref{MIMO_ILC:eq:ILC_mon_conv_iinfty_SISO} follows from \eqref{MIMO_ILC:eq:ILC_mon_conv_cond_rho}, for both $p=1$, $p=\infty$.	
\end{proof}

\begin{remark}
	\thmref{MIMO_ILC:thm:ILC_decentralized_Gershgorin} presents the ILC-analog of Gershgorin bounds in feedback control, see, e.g., \cite{Grosdidier1986}, \cite[Section 10.6]{Skogestad2007}.
\end{remark}

\subsubsection{Independent \texorpdfstring{$Q$}{Q}-filter Design Based on the SSV} \label{MIMO_ILC:subsec:dec_ILC_indep_Q_factorized}
Alternatively, conditions are developed using the structured singular value (SSV), see, e.g., \cite{Grosdidier1986,Zhou1996}. The idea is to exploit the diagonal structure of $QM_\mathrm{d}$ in \lemref{MIMO_ILC:lem:factorization_rho_conditions}.
Particularly, for a matrix $A$ and diagonal matrix $\Delta$, see \cite[eq. (8.95)]{Skogestad2007}, it holds 
\begin{equation}\label{MIMO_ILC:eq:rho_upper_bound}
\rho(\Delta A)\leq \bar{\sigma}(\Delta)\mu_\Delta(A),
\end{equation}
where $\mu_\Delta(A)$ is taken with respect to the structure of $\Delta$.
\begin{definition}
	For $A\in\mathbb{C}^{n\times n}$, the SSV $\mu_\Delta(A)$ is defined
	\begin{equation}
		\mu_\Delta(A)=\tfrac{1}{\vphantom{\tilde{A}}\min\{\bar{\sigma}(\Delta):\Delta\in\mathbf{\Delta},\det(I-A\Delta)=0\}},
	\end{equation}
	where $\mathbf{\Delta}$ is a prescribed set of block diagonal matrices, unless no $\Delta\in\mathbf{\Delta}$ makes $I-A\Delta$ singular, in which case $\mu_\Delta(A)=0$.
\end{definition}
\begin{theorem}\label{MIMO_ILC:thm:ILC_decentralized_SSV}
	Iteration \eqref{MIMO_ILC:eq:ILC_feedforward_propagation} with decentralized filter $Q(z)=\mathrm{diag}\{q_{i}(z)\}$, as in \eqref{MIMO_ILC:eq:LQ_decentralized}, converges if
	\begin{align}
		|q_i(e^{\iota\omega})M_{ii}(e^{\iota\omega})| < \tfrac{1}{\vphantom{\tilde{A}} \mu_\mathrm{d}(I+E(e^{\iota\omega}))} \quad\forall i, \omega\in[0,\pi] ,
		\label{MIMO_ILC:eq:ILC_conv_mu_SISO}
	\end{align}
	and converges monotonically w.r.t. the $\ell_2$ norm of $f_j$ if
	\begin{align}
		|q_i(e^{\iota\omega})M_{ii}(e^{\iota\omega})| < \tfrac{1}{\vphantom{\tilde{A}}\sqrt{\mu_\mathrm{d}((I+E(e^{\iota\omega}))(I+E(e^{\iota\omega}))^H)}} \\ \quad\forall i, \omega\in[0,\pi],\hspace{10pt}
		\label{MIMO_ILC:eq:ILC_mon_conv_mu_SISO}
	\end{align}
	where $\mu_\mathrm{d}(\cdot)$ is the SSV with respect to a diagonal structure.
\end{theorem}
\begin{proof}
	Using \eqref{MIMO_ILC:eq:rho_upper_bound} in \eqref{MIMO_ILC:eq:ILC_conv_cond_rho}, where $\Delta=QM_\mathrm{d}$, $A=I+E$ and $\mathbf{\Delta}=\{\delta I:\delta\in\mathbb{C}^n\}$, gives
	\begin{equation}\label{MIMO_ILC:eq:ILC_conv_factorized}
		\rho(Q(e^{\iota\omega})M_\mathrm{d}(e^{\iota\omega}))  < \tfrac{1}{\vphantom{\tilde{A}}\mu_\mathrm{d}(I+E(e^{\iota\omega}))}
		\quad\forall \omega\in[0,\pi].
	\end{equation}
	Omitting arguments, $\rho(QM_\mathrm{d})=\max_i |q_iM_{ii}|$ implies \eqref{MIMO_ILC:eq:ILC_conv_mu_SISO}.
	Applying \eqref{MIMO_ILC:eq:rho_upper_bound} to \eqref{MIMO_ILC:eq:ILC_mon_conv_cond_rho}, observing that $\bar{\sigma}(M_\mathrm{d}^HQ^HQM_\mathrm{d})=\max_i|q_iM_{ii}|^2$, and taking square roots proves \eqref{MIMO_ILC:eq:ILC_mon_conv_mu_SISO}.
\end{proof}
The SSV is employed in a fundamentally different way than in stability analyses of feedback systems. In robust control, e.g., \cite[Chapters 9, 11]{Zhou1996}, typically $\mu_\Delta(M)$ is taken with respect to structured uncertainty $\Delta$, and $M$ denotes a nominal model. In contrast, here $I+E$ has the role of nominal model, and $QM_\mathrm{d}$ is the structured uncertainty yet to be designed. 

\subsection{Decentralized \texorpdfstring{$Q$}{Q}-filter Design for Robustness to Interaction}\label{MIMO_ILC:subsec:decentralized_robust_ILC_design}
Theorems \ref{MIMO_ILC:thm:ILC_decentralized_Gershgorin} and  \ref{MIMO_ILC:thm:ILC_decentralized_SSV} enable systematic and robust (\ref{MIMO_ILC:R1}) decentralized design, using only SISO parametric models (\ref{MIMO_ILC:R3}). 
This is summarized as follows, constituting step \circled{5} of \procref{MIMO_ILC:proc_MIMO_ILC}.
\begin{algorithm}\label{MIMO_ILC:algo_decentralized_ILC}
	{\normalsize Step \circled{5}: robust decentralized MIMO design}\vspace{-2mm}
	\par\noindent\hrulefill
	\begin{enumerate}
		[label=\alph*)]
		\item Obtain SISO parametric models $\widehat{J_{ii}}(z)$ of $J_{ii}(z)$, $i=1,\ldots,n$.
		\item Design $L(z)=\mathrm{diag}\{l_{i}(z)\}$ such that $1-l_{i}\widehat{J_{ii}}\approx 0$,  $\forall i$.
		\item Construct $M_\mathrm{d}(e^{\iota\omega})$ and $E(e^{\iota\omega})$, see \eqref{MIMO_ILC:eq:ILC_f_propagation}, based on $L$ from b) and MIMO FRF model $\hat{J}_\textrm{FRF}(e^{\iota\omega})$ from step \circled{1} of \procref{MIMO_ILC:proc_MIMO_ILC}.
		\item For robust stability, design $Q(z)=\mathrm{diag}\{q_i\}$ according to joint evaluation of Theorems \ref{MIMO_ILC:thm:ILC_decentralized_Gershgorin}, \ref{MIMO_ILC:thm:ILC_decentralized_SSV}, i.e., for each frequency $\omega\in[0,\pi]$,
		\begin{itemize}
			\item at least one of \eqref{MIMO_ILC:eq:ILC_conv_iinfty_SISO}, \eqref{MIMO_ILC:eq:ILC_conv_i1_SISO}, \eqref{MIMO_ILC:eq:ILC_conv_mu_SISO} is satisfied (convergence);
			\item at least one of \eqref{MIMO_ILC:eq:ILC_mon_conv_iinfty_SISO}, \eqref{MIMO_ILC:eq:ILC_mon_conv_mu_SISO} is satisfied (monotonic convergence).
		\end{itemize}
	\end{enumerate}
\end{algorithm}

The key advantage of \algoref{MIMO_ILC:algo_decentralized_ILC}, compared to \algoref{MIMO_ILC:algo_robust_multi-loop_ILC}, is that performance (\ref{MIMO_ILC:R2}) can potentially be increased, while the modeling requirements (\ref{MIMO_ILC:R3}) remain equal.
Indeed, \eqref{MIMO_ILC:eq:ILC_conv_iinfty_SISO}-\eqref{MIMO_ILC:eq:ILC_mon_conv_iinfty_SISO} and \eqref{MIMO_ILC:eq:ILC_conv_mu_SISO}-\eqref{MIMO_ILC:eq:ILC_mon_conv_mu_SISO} can be computed using $\hat{J}_\textrm{FRF}(e^{\iota\omega})$, such that interaction does not have to be included in models $\widehat{J_{ii}}(z)$.

\begin{remark}\label{MIMO_ILC:rem:ILC_decentralized_joint_conditions}
	In \algoref{MIMO_ILC:algo_decentralized_ILC}, the developed bounds \eqref{MIMO_ILC:eq:ILC_conv_iinfty_SISO}, \eqref{MIMO_ILC:eq:ILC_conv_i1_SISO}, \eqref{MIMO_ILC:eq:ILC_conv_mu_SISO}, respectively \eqref{MIMO_ILC:eq:ILC_mon_conv_iinfty_SISO}, \eqref{MIMO_ILC:eq:ILC_mon_conv_mu_SISO}, are jointly considered. This is since the ordering of their tightness may vary as a function of frequency, and hence they all contribute to the design. Note however that, for a specific frequency $\omega\in[0,\pi]$, they can in general not be combined over the different SISO loops $i$. That is, convergence is guaranteed only if, per evaluated frequency, at least one condition is satisfied for all loops $i$ simultaneously.
\end{remark}
\begin{remark}
	In the SISO case, the results in Theorems \ref{MIMO_ILC:thm:ILC_decentralized_Gershgorin} and \ref{MIMO_ILC:thm:ILC_decentralized_SSV} recover the SISO condition \eqref{MIMO_ILC:eq:convergence_decentralized}, since in this case $E=0$.
\end{remark}

The achievable performance of decentralized ILC, i.e., the magnitude of $e_\infty$, is limited by interaction that is ignored in the design of $L$. If increased modeling effort is justified (\ref{MIMO_ILC:R3}) in relation to performance requirements (\ref{MIMO_ILC:R2}), MIMO parametric models of $J$ can be used to design centralized ILC algorithms.

\section{Centralized ILC: Accounting for Interaction Through \texorpdfstring{$L$}{L}-Filter Designs}\label{MIMO_ILC:sec:centralized}
For systems where decentralized ILC yields unsatisfactory performance due to required robustness to ignored interaction, and increased modeling effort (\ref{MIMO_ILC:R3}) is justified in view of performance requirements (\ref{MIMO_ILC:R2}), centralized ILC schemes can be designed using a full MIMO model, i.e., step \circled{6}.
By explicitly accounting for interaction in $L$, the requirement for robustness through $Q$ is alleviated, which potentially increases performance.
This leads to the following algorithm.
\begin{algorithm}\label{MIMO_ILC:algo_centralized_ILC}
	{\normalsize Step \circled{6}: centralized MIMO design}\vspace{-2mm}
	\par\noindent\hrulefill
	\begin{enumerate}
		[label=\alph*)]
		\item Obtain MIMO parametric model $\hat{J}(z)$, including interaction;
		\item Design $L(z)$ such that $I-L\hat{J} \approx 0$, e.g., \cite{Zundert2018Mech,Blanken2016CDC,Blanken2016Mech}.
		\item For robust stability, design $Q(z)$ according to \corref{MIMO_ILC:cor:ILC_conv_qd}, Theorem \ref{MIMO_ILC:thm:ILC_decentralized_Gershgorin} or \ref{MIMO_ILC:thm:ILC_decentralized_SSV} based on FRF model $\hat{J}_\textrm{FRF}(e^{\iota\omega})$ from step \circled{1}.
	\end{enumerate}
\end{algorithm}

In preceding sections, the techniques underlying steps \circled{4} to \circled{6} are developed. Next, these are applied to the case study.

\section{Application of Design Framework to Multivariable Case Study}
\label{MIMO_ILC:sec:case_study}
In this section, \procref{MIMO_ILC:proc_MIMO_ILC} is applied to a case study in a step-by-step manner.
Simulations are performed to clearly show the differences between the developed approaches.
Details on the case study and Matlab implementations of \procref{MIMO_ILC:proc_MIMO_ILC} are available as Supplementary Material.

\subsection{Case Study: Oc\'e Arizona 550GT Flatbed Printer}\label{MIMO_ILC:subsec:case_study}
\begin{figure}[tpb]
	\centering
	\mbox{\includegraphics[scale=0.8,page=2]{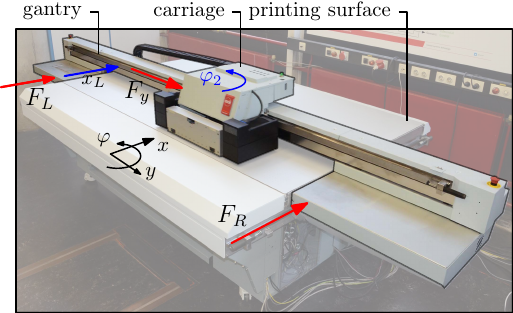}}
	\caption{Oc\'e Arizona 550GT flatbed printer. The carriage moves along the gantry, which provides the motion freedom to cover the printing surface. The actuator forces are indicated by red arrows. The inputs considered for control are $F_L$, $F_R$, and the outputs are $x_L$, $\varphi_2$, indicated by blue arrows.}
	\label{MIMO_ILC:fig:Arizona_CS}
\end{figure}

An Oc\'e Arizona 550GT printer is considered, see \figref{MIMO_ILC:fig:Arizona_CS}.
In contrast to standard consumer printers, the printer can print on both flexible and rigid media, e.g., paper, plastics, wood and metals. The medium is fixed on the printing surface, and the carriage, which contains the printheads, moves in the horizontal plane.
This yields inherently multivariable dynamics.

The simulations are performed using the model shown in \figref{MIMO_ILC:fig:Bode_system}.
The inputs are forces $F_L$ $[\mathrm{N}]$ and $F_R$ $[\mathrm{N}]$ acting on the gantry; the outputs are the gantry position at the left side $x_L$ $[\mathrm{m}]$, and carriage rotation $\varphi_2$ $[\mathrm{rad}]$, i.e., $\begin{bmatrix}x_L \\ \varphi_2\end{bmatrix}=G_o\begin{bmatrix}F_L \\ F_R\end{bmatrix}$, where $G_o$ is the system before decoupling in step \circled{3}.

\begin{table*}[!tb]
		\caption{Overview of approaches in design framework: modeling requirements, design parameters, and asymptotic performance.}
		\label{MIMO_ILC:tab:overview}
		\begin{center}
			\setlength\tabcolsep{4.5pt}			
			\setlength\extrarowheight{1pt}
			\begin{tabular}{c|c|c|c|c|c|c|x{8mm}|x{8mm}|c|c}
				& \multirow{3}{*}{$L$-filter} & \multirow{3}{1.4cm}{\centering required parametric models} & \multirow{3}{*}{$Q$-filter} & \multirow{3}{1.6cm}{\centering guaranteed robust stability (\ref{MIMO_ILC:R1})} & \multirow{3}{1.5cm}{\centering performance (\ref{MIMO_ILC:R2})} & \multirow{3}{1.2cm}{\centering user effort (\ref{MIMO_ILC:R3})} & \multicolumn{2}{c|}{cut-off $f_c~[\mathrm{Hz}]$} & \multirow{3}{1.1cm}{ $\|e_\infty\|_F$} & \\
				\cline{8-9} & & & & & & & \multirow{2}{*}{$q_1$} & \multirow{2}{*}{$q_2$} & & \\
				& & & & & & & & & & \\ \hline\hline
				Proc. \ref{MIMO_ILC:proc_SISO_ILC} & $n$ $\times$ SISO & $n$ $\times$ SISO& SISO, \S\ref{MIMO_ILC:subsec:SISO_design} & no & $-$ & $+$ & 100 & 15 &  N/A & \protect\tikz[baseline=-0.6ex,x=1pt,y=1pt]{\protect\draw[black,thick] [-] (0,-3) -- (6,3);\protect\draw[black,thick] [-] (0,3) -- (6,-3);} \\ \hline
				Proc. \ref{MIMO_ILC:proc_MIMO_ILC}: \circled{4} & $n$ $\times$ SISO & $n$ $\times$ SISO & Robust SISO, \S\ref{MIMO_ILC:subsec:Q_robustness_interaction} & yes & $-$ & $+$ & 13 & 13 & $0.45$ & \protect\tikz[baseline=-0.6ex,x=1pt,y=1pt]{\protect\draw[red,thick] (0,-3) rectangle ++(6,6);} \\ \hline
				Proc. \ref{MIMO_ILC:proc_MIMO_ILC}: \circled{5} & $n$ $\times$ SISO & $n$ $\times$ SISO & Decentralized, \S\ref{MIMO_ILC:sec:decentralized} & yes & $+$/$-$& $+$/$-$ & 19 & 9 & $0.30$ & \protect\tikz[baseline=-0.6ex,x=1pt,y=1pt]{\protect\draw[blue,thick] (7,0) circle (3);} \\ \hline
				Proc. \ref{MIMO_ILC:proc_MIMO_ILC}: \circled{6} & MIMO, \S\ref{MIMO_ILC:sec:centralized} & full MIMO & Robust SISO, \S\ref{MIMO_ILC:subsec:Q_robustness_interaction} & yes & $+$ & $-$& 28 & 28 &  $0.14$ & \protect\tikz[baseline=-0.6ex,x=1pt,y=1pt]{\protect\draw[green,thick] (0,0) +(3,0) -- +(0,3) -- +(-3,0) -- +(0,-3) -- cycle;}
			\end{tabular}
		\end{center}
\end{table*}

The system is discretized using zero-order-hold on the input with sampling interval $10^{-3}$ s.
A stabilizing diagonal feedback controller $C(z)=\mathrm{diag}\{c_1(z),c_2(z)\}$ is designed, where
\begin{equation}
c_1(z) = \frac{5\times10^4(z-0.988)}{z-0.939},~
c_2(z) = \frac{1.3\times10^4(z-0.991)}{z-0.969},
\end{equation}
yielding a bandwidth of $3$ $\mathrm{Hz}$ in $x_L$ direction and $1.5$ $\mathrm{Hz}$ in $\varphi_2$ direction.
The system $J(z)=(I+GC)^{-1}G$ has non-minimum phase transmission zeros at $z=1.09$ and $z=-6.69$ due to the non-collocated inputs/outputs and fast sampling. 
A model $\hat{J}(z)=(I+\hat{G}C)^{-1}\hat{G}$ is provided for ILC design, see \figref{MIMO_ILC:fig:Bode_system}.
A modeling error is present at the first resonance in the $(2,2)$-element, which plays a crucial role in the designs.

\begin{figure}[!tb]
	\centering
	\mbox{\includegraphics{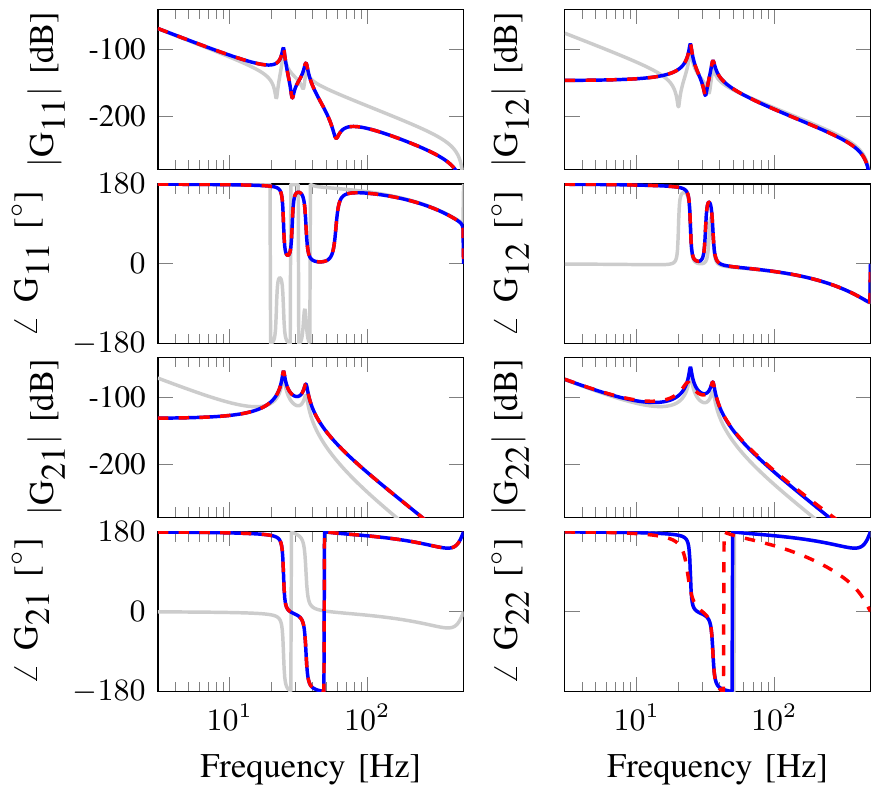}}
	\caption{Bode diagram of non-decoupled true plant $G_o(z)$ (\protect\tikz[baseline=-0.6ex,x=1pt,y=1pt]{ \protect\draw[gray,thick] [-] (0,0) -- (10,0);}), true plant $G$ (\protect\tikz[baseline=-0.6ex,x=1pt,y=1pt]{ \protect\draw[blue,thick] [-] (0,0) -- (10,0);}) after decoupling transformations in step \protect\circled{3} of \procref{MIMO_ILC:proc_MIMO_ILC}, and model of decoupled plant $\widehat{G}(z)$ (\protect\tikz[baseline=-0.6ex,x=1pt,y=1pt]{ \protect\draw[red,thick] [-] (0,0) -- (4,0);\protect\draw[red,thick] [-] (6,0) -- (10,0);}) used for ILC design.}
	\label{MIMO_ILC:fig:Bode_system}
\end{figure}

\subsection{Results: Application of \procref{MIMO_ILC:proc_MIMO_ILC} to Case Study}

\begin{figure}[!tb]
		\centering
		\mbox{\includegraphics{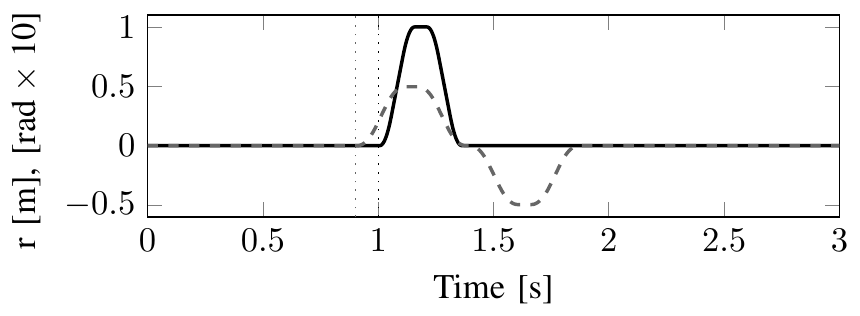}}
		\caption{Reference trajectories $r_x$ (\protect\tikz[baseline=-0.6ex,x=1pt,y=1pt]{ \protect\draw[black,thick] [-] (0,0) -- (10,0);}) and $r_\varphi$ (\protect\tikz[baseline=-0.6ex,x=1pt,y=1pt]{ \protect\draw[gray,thick] [-] (0,0) -- (4,0);\protect\draw[gray,thick] [-] (6,0) -- (10,0);}). The start of the motion tasks are indicated by dotted lines.}
		\label{MIMO_ILC:fig:references}
\end{figure}

Next, \procref{MIMO_ILC:proc_MIMO_ILC} is step-by-step applied to the case study, and the results are presented.
The disturbance $r=[r_x,r_\varphi]^\top$ of length $N=3001$ is shown in \figref{MIMO_ILC:fig:references}.
An overview of the designs is provided in \tabref{MIMO_ILC:tab:overview}, and the resulting performance $\|e_j\|_F$ is shown in \figref{MIMO_ILC:fig:norm_error}, where $\|e_j\|_F=\sqrt{\sum_{i,k}|e_j(i,k)|^2}$ with $e_j=[e_{j,x},e_{j,\varphi}]^\top\in\mathbb{R}^{N\times2}$. Note that $e_{x}$ [m] and $e_{\varphi}$ [rad] are weighed equally since they have comparable magnitude.
\begin{enumerate}[label=\protect\circled{\arabic*},leftmargin=0cm,itemindent=2\parindent]
	\item \underline{Non-parametric modeling}: it is assumed that the MIMO non-parametric FRF model is exact, i.e., $\hat{J}_\textrm{FRF}(e^{\iota\omega})=J(e^{\iota\omega})$.
	
	\item \underline{Interaction analysis}.
	From \figref{MIMO_ILC:fig:Bode_system}, it can be directly observed that there is substantial interaction above $20$ Hz.
	
	\item \underline{Decoupling transformations}.
	The input is transformed using static matrix $T_u$, see, e.g., \cite{Oomen2018}, such that plant $G=G_oT_u$ is diagonally dominant at low frequencies, see \figref{MIMO_ILC:fig:Bode_system}.
	
	\item \underline{Robust multi-loop SISO design}. The filters $l_{i}=1/\hat{J}_{ii}$ are implemented using stable inversion, see, e.g., \cite{Zundert2018Mech}.
	Filters $q_i$ are first-order zero-phase low-pass Butterworth filters.
	\begin{itemize}
		\item Independent SISO schemes (\protect\tikz[baseline=-0.6ex,x=1pt,y=1pt]{\protect\draw[black,thick] [-] (0,-3) -- (6,3);\protect\draw[black,thick] [-] (0,3) -- (6,-3);},\protect\tikz[baseline=-0.6ex,x=1pt,y=1pt]{ \protect\draw[black,thick] [-] (0,0) -- (10,0);}) are designed according to \eqref{MIMO_ILC:eq:convergence_decentralized} and \procref{MIMO_ILC:proc_SISO_ILC}, see \figref{MIMO_ILC:subfig:SISO_bounds}.
		Since interaction is ignored, convergence is not guaranteed (\ref{MIMO_ILC:R1}), see \figref{MIMO_ILC:subfig:SISO_convergence}. This is corroborated by \figref{MIMO_ILC:fig:norm_error}.
		\item Through robust SISO design (\protect\tikz[baseline=-0.6ex,x=1pt,y=1pt]{\protect\draw[red,thick] (0,-3) rectangle ++(6,6);},\protect\tikz[baseline=-0.6ex,x=1pt,y=1pt]{ \protect\draw[red,thick] [-] (0,0) -- (10,0);}) according to \algoref{MIMO_ILC:algo_robust_multi-loop_ILC}, convergence is guaranteed using MIMO FRF model $\hat{J}_\textrm{FRF}(e^{\iota\omega})$, see \figref{MIMO_ILC:subfig:SISO_convergence}.
		Note that $q_d$ is cut off at a low frequency due to required robustness in loop 2. This comes at the cost of performance (\ref{MIMO_ILC:R2}), also in loop 1, see \figref{MIMO_ILC:fig:error}.
	\end{itemize}
	
	\item \underline{Decentralized robust MIMO design} (\protect\tikz[baseline=-0.6ex,x=1pt,y=1pt]{\protect\draw[blue,thick] (7,0) circle (3);},\protect\tikz[baseline=-0.6ex,x=1pt,y=1pt]{ \protect\draw[blue,thick] [-] (0,0) -- (10,0);}) using \algoref{MIMO_ILC:algo_decentralized_ILC} further improves performance, see \figref{MIMO_ILC:fig:norm_error}.
	The same models are used: only decentralized filter $Q(z)=\mathrm{diag}\{q_i(z)\}$ is designed in a more sophisticated manner, see \figref{MIMO_ILC:subfig:DEC_convergence}.
	\begin{itemize}
		\item Compared to robust SISO design in \circled{4}, the cut-off frequency of $q_1$ is significantly higher, see \tabref{MIMO_ILC:tab:overview}.
		In loop 2, the modeling error is dominant beyond $10$ Hz, whereas in loop 1 robustness is required to interaction above $20$ Hz.
		
		\item The main improvement is achieved in loop 1, see \figref{MIMO_ILC:fig:error}, whereas the error in loop 2 is slightly increased.
		
		\item In view of d) of \algoref{MIMO_ILC:algo_decentralized_ILC}, note that each condition \eqref{MIMO_ILC:eq:ILC_conv_iinfty_SISO}, \eqref{MIMO_ILC:eq:ILC_conv_i1_SISO}, \eqref{MIMO_ILC:eq:ILC_conv_mu_SISO}, is violated at least once over the frequency range, see \figref{MIMO_ILC:subfig:DEC_convergence}. Hence, they all contribute to the design.
	\end{itemize}
	
	\item \underline{Centralized MIMO design} (\protect\tikz[baseline=-0.6ex,x=1pt,y=1pt]{\protect\draw[green,thick] (0,0) +(3,0) -- +(0,3) -- +(-3,0) -- +(0,-3) -- cycle;},\protect\tikz[baseline=-0.6ex,x=1pt,y=1pt]{ \protect\draw[green,thick] [-] (0,0) -- (10,0);}) using \algoref{MIMO_ILC:algo_centralized_ILC} yields the highest performance. 
	Given MIMO model $\hat{J}$, the MIMO filter $L=\hat{J}^{-1}$ is implemented using stable inversion.
	\begin{itemize}
		\item By designing for interaction in $L$, less robustness is required compared to \circled{4} and \circled{5}: the cut-off frequency of $Q=q_dI$, see \corref{MIMO_ILC:cor:ILC_conv_qd}, is significantly higher, see \figref{MIMO_ILC:subfig:CEN_stabinv_convergence}.
		
		\item The required MIMO model $\hat{J}$ can be expensive to obtain in practice: user effort (\ref{MIMO_ILC:R3}) is sacrificed for performance (\ref{MIMO_ILC:R2}).
		
		\item All approaches generate pre-actuation to compensate the non-minimum phase transmission zeros of $J$, see \figref{MIMO_ILC:fig:feedforward}.
	\end{itemize}
\end{enumerate}

\begin{figure}[!tb]
	\centering
	{\includegraphics{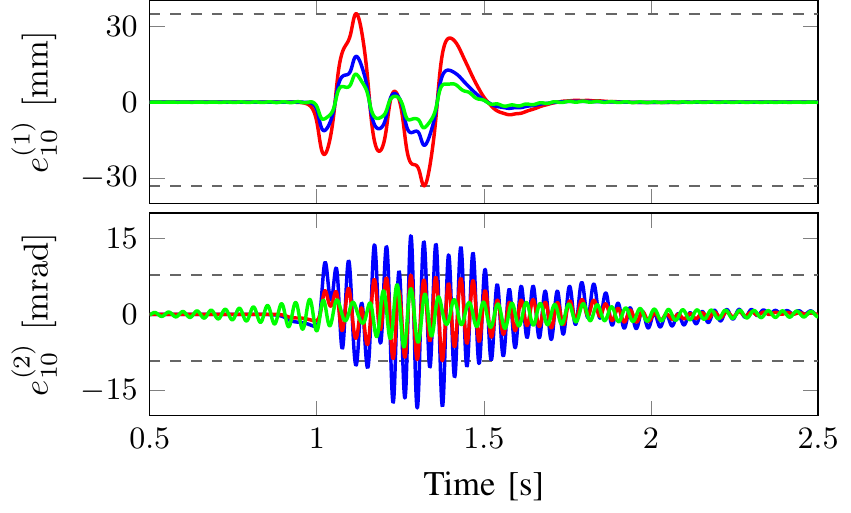}}
	\caption{Error signals in trial $j=10$ of robust SISO design in step \protect\circled{4} (\protect\tikz[baseline=-0.6ex,x=1pt,y=1pt]{ \protect\draw[red,thick] [-] (0,0) -- (10,0);}), decentralized design in step \protect\circled{5} (\protect\tikz[baseline=-0.6ex,x=1pt,y=1pt]{ \protect\draw[blue,thick] [-] (0,0) -- (10,0);}) and centralized design in step \protect\circled{6} (\protect\tikz[baseline=-0.6ex,x=1pt,y=1pt]{ \protect\draw[green,thick] [-] (0,0) -- (10,0);}).}
	\label{MIMO_ILC:fig:error}
\end{figure}
\begin{figure}[!tb]
	\centering
	{\includegraphics{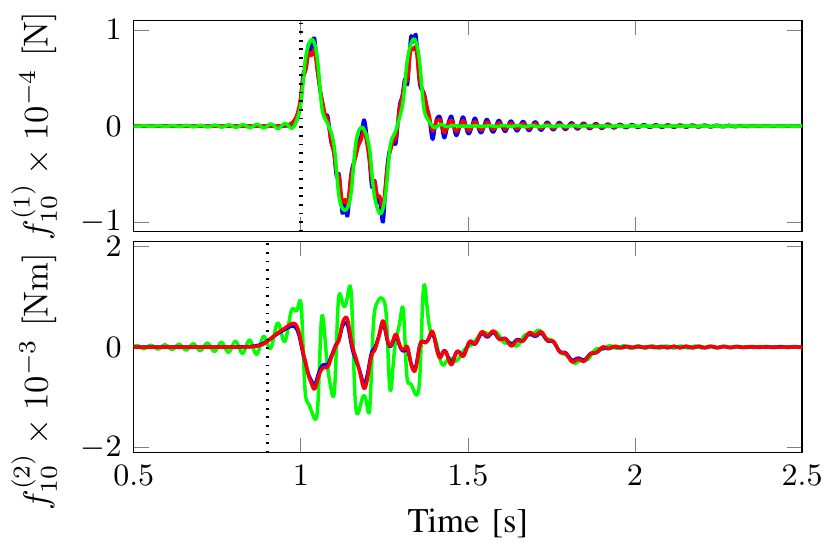}}
	\caption{Feedforward signals in trial $j=10$ of the robust SISO design in step \protect\circled{4} (\protect\tikz[baseline=-0.6ex,x=1pt,y=1pt]{ \protect\draw[red,thick] [-] (0,0) -- (10,0);}), decentralized design in step \protect\circled{5} (\protect\tikz[baseline=-0.6ex,x=1pt,y=1pt]{ \protect\draw[blue,thick] [-] (0,0) -- (10,0);}) and centralized design in step \protect\circled{6} (\protect\tikz[baseline=-0.6ex,x=1pt,y=1pt]{ \protect\draw[green,thick] [-] (0,0) -- (10,0);}). The start of the motion tasks are indicated by dotted lines.}
	\label{MIMO_ILC:fig:feedforward}
\end{figure}

The following key conclusions are made: i) interaction must be taken into account in the design, ii) performance can be improved with limited user effort through decentralized designs, and iii) if justified by performance requirements, performance can be further improved through centralized MIMO design.

\section{Conclusions}\label{MIMO_ILC:sec:conclusions}
The design framework developed in this paper enables systematic design of ILC controllers for multivariable systems, and balances performance requirements with modeling and design effort through a range of design solutions. This is done by judiciously combining non-parametric FRF measurements and parametric models.
The results are demonstrated on a flatbed printing system, including trade-offs between approaches.

\begin{figure}[tpb]
	\centering
	\subfigure[\vspace{-1pt}Interaction-ignoring SISO design using \eqref{MIMO_ILC:eq:convergence_decentralized}: filters $q_{i}$ ({\protect\tikz[baseline=-0.6ex,x=1pt,y=1pt]{ \protect\draw[black,thick] [-] (0,0) -- (3.3,0);\protect\draw[black,thick] [-] (4.7,0) -- (5.3,0);\protect\draw[black,thick] [-] (6.8,0) -- (10,0);}}) (top: $i=1$; bottom: $i=2$) are designed such that $|q_{i}(1-l_{i}\hat{J}_\textrm{FRF,ii})|<1,\forall \omega$ ({\protect\tikz[baseline=-0.6ex,x=1pt,y=1pt]{ \protect\draw[black,thick] [-] (0,0) -- (10,0);}}).\vspace{-2pt} \label{MIMO_ILC:subfig:SISO_bounds}]	{\mbox{\includegraphics[scale=1]{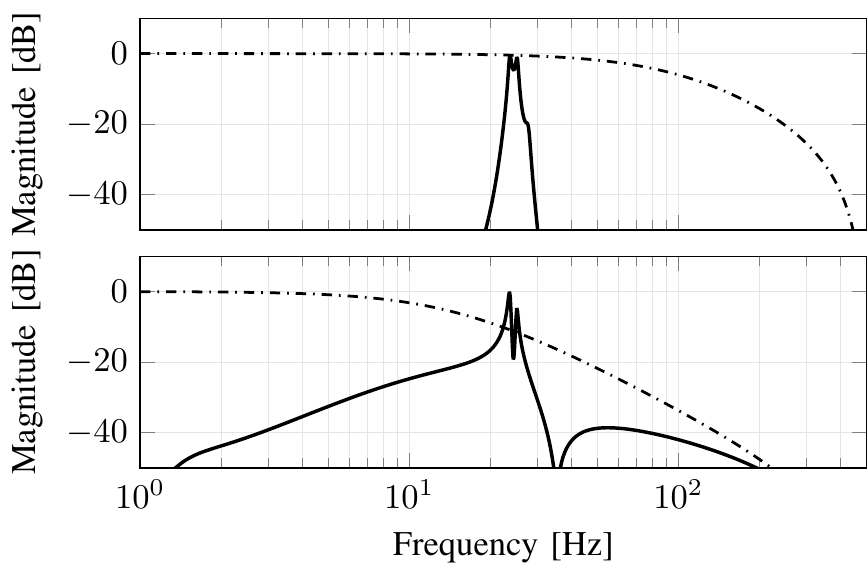}}}
	\\
	\subfigure[\vspace{-1pt}Step {\protect\circled{4}}: robust SISO design of $Q=q_dI$ ({\protect\tikz[baseline=-0.6ex,x=1pt,y=1pt]{ \protect\draw[red,thick] [-] (0,0) -- (3.3,0);\protect\draw[red,thick] [-] (4.7,0) -- (5.3,0);\protect\draw[red,thick] [-] (6.8,0) -- (10,0);}}) according to \algoref{MIMO_ILC:algo_robust_multi-loop_ILC} guarantees convergence, i.e., $\rho(Q(I-L\hat{J}_\textrm{FRF}))<1,\forall \omega$ ({\protect\tikz[baseline=-0.6ex,x=1pt,y=1pt]{ \protect\draw[red,thick] [-] (0,0) -- (10,0);}}), in contrast to interaction-ignoring SISO designs ({\protect\tikz[baseline=-0.6ex,x=1pt,y=1pt]{ \protect\draw[black,thick] [-] (0,0) -- (10,0);}}), see \figref{MIMO_ILC:subfig:SISO_bounds}.\vspace{-2pt} \label{MIMO_ILC:subfig:SISO_convergence}]
	{\mbox{\includegraphics[scale=1]{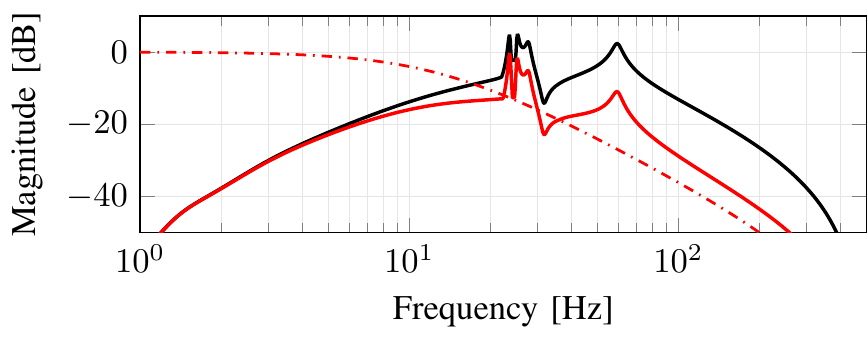}}}
	\\
	\subfigure[\vspace{-1pt}Step {\protect\circled{5}}: robust decentralized design according to \algoref{MIMO_ILC:algo_decentralized_ILC}: the filters $q_i$ ({\protect\tikz[baseline=-0.6ex,x=1pt,y=1pt]{ \protect\draw[blue,thick] [-] (0,0) -- (3.3,0);\protect\draw[blue,thick] [-] (4.7,0) -- (5.3,0);\protect\draw[blue,thick] [-] (6.8,0) -- (10,0);}}) (top: $i=1$; bottom: $i=2$) are chosen such that, for each frequency, $|q_i(1-l_i\hat{J}_\textrm{FRF,ii})|=|q_iM_{ii}|$ ({\protect\tikz[baseline=-0.6ex,x=1pt,y=1pt]{ \protect\draw[blue,thick] [-] (0,0) -- (10,0);}}) are upper bounded by at least one of the right-hand sides of \eqref{MIMO_ILC:eq:ILC_conv_iinfty_SISO}, \eqref{MIMO_ILC:eq:ILC_conv_i1_SISO} and \eqref{MIMO_ILC:eq:ILC_conv_mu_SISO} ({\protect\tikz[baseline=-0.6ex,x=1pt,y=1pt]{ \protect\draw[black,thick] [-] (0,0) -- (4,0);\protect\draw[black,thick] [-] (6,0) -- (10,0);}, \protect\tikz[baseline=-0.6ex,x=1pt,y=1pt]{ \protect\draw[green,thick] [-] (0,0) -- (4,0);\protect\draw[green,thick] [-] (6,0) -- (10,0);}, \protect\tikz[baseline=-0.6ex,x=1pt,y=1pt]{ \protect\draw[red,thick] [-] (0,0) -- (4,0);\protect\draw[red,thick] [-] (6,0) -- (10,0);}, respectively}).\vspace{-2pt} \label{MIMO_ILC:subfig:DEC_convergence}]	
	{\mbox{\includegraphics[scale=1]{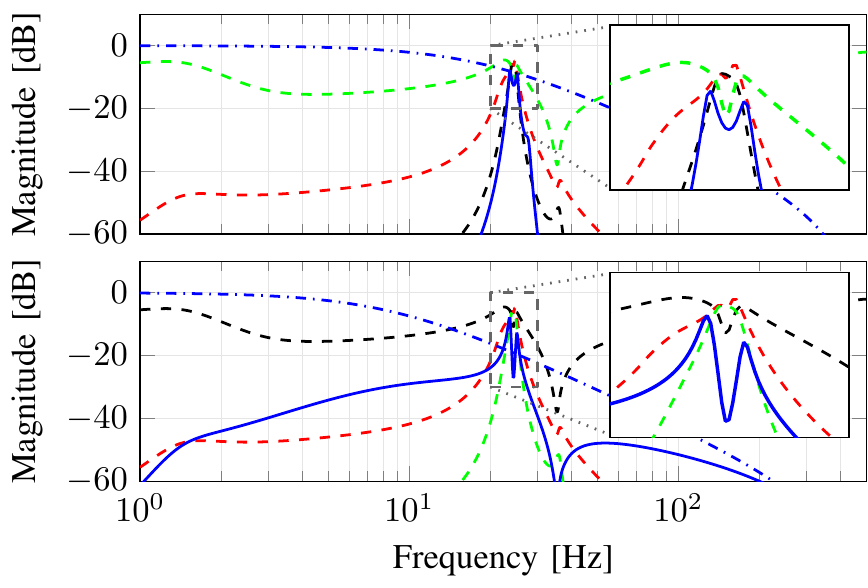}}}	
	\\
	\subfigure[\vspace{-1pt}Step {\protect\circled{6}}: centralized ILC design according to \algoref{MIMO_ILC:algo_centralized_ILC} and \eqref{MIMO_ILC:eq:conv_cond_qd}: the filter $Q=q_dI$ ({\protect\tikz[baseline=-0.6ex,x=1pt,y=1pt]{ \protect\draw[green,thick] [-] (0,0) -- (3.3,0);\protect\draw[green,thick] [-] (4.7,0) -- (5.3,0);\protect\draw[green,thick] [-] (6.8,0) -- (10,0);}}) is designed such that $|q_d|\rho(1-L\hat{J}_\textrm{FRF})<1,\forall \omega$ ({\protect\tikz[baseline=-0.6ex,x=1pt,y=1pt]{ \protect\draw[green,thick] [-] (0,0) -- (10,0);}}).\vspace{-12pt} \label{MIMO_ILC:subfig:CEN_stabinv_convergence}]	{\mbox{\includegraphics[scale=1]{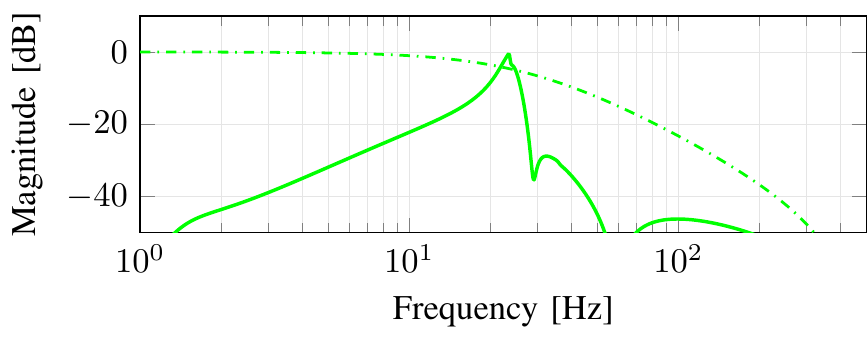}}}			
	\caption{Application of \procref{MIMO_ILC:proc_MIMO_ILC}: designs of robustness filters $Q(z)$.}
	\label{MIMO_ILC:fig:designs_Q}
\end{figure}


\section*{Acknowledgment}
The authors thank Sjirk Koekebakker, Maarten Steinbuch, Jeroen Willems and Jurgen van Zundert for their contributions.


\end{document}